\documentclass[10pt,journal,twocolumn,twoside]{autart}    %

\usepackage{graphicx}
\usepackage{amsmath}
\usepackage{xcolor}
\usepackage{amssymb}
\usepackage{mathtools}
\usepackage{booktabs}
\usepackage{multirow}
\usepackage[hidelinks]{hyperref}
\let\theoremstyle\relax
\usepackage{amsthm}
\usepackage[shortlabels]{enumitem}
\usepackage[T1]{fontenc}
\usepackage[utf8]{inputenc}

\makeatletter
\let\cl@part\undefined
\makeatletter

\allowdisplaybreaks

\usepackage[capitalise]{cleveref} %
\usepackage{subcaption}

\newif\ifautomatica

\automaticafalse

\newtheorem{lemma}{Lemma}
\newtheorem{assumption}{Assumption}

\newtheorem{definition}{Definition}
\newtheorem{corollary}{Corollary}

\newtheorem{theorem}{Theorem}
\newtheorem{remark}{Remark}

\usepackage{amsfonts,bm}

\ifautomatica
\newcommand{\revision}[1]{\textcolor{blue}{#1}}
\else
\newcommand{\revision}[1]{\textcolor{black}{#1}}
\fi

\def\eqref#1{equation~\ref{#1}}

\def\1{\bm{1}}

\DeclareMathAlphabet{\mathsfit}{\encodingdefault}{\sfdefault}{m}{sl}
\SetMathAlphabet{\mathsfit}{bold}{\encodingdefault}{\sfdefault}{bx}{n}

\newcommand{\E}{\mathbb{E}}

\usepackage{comment}

\begin{document}

\begin{frontmatter}

\title{Stochastic Model Predictive Control for Sub-Gaussian Noise}

\thanks[footnoteinfo]{This work is in part supported by the Hasler Foundation ("Learn to learn safely" project, grant number: 21039), and Swiss National Science Foundation under NCCR Automation, grant agreement 51NF40 180545, and ETH AI Center.
}

\author[balgrist,eth_cs,eth_ai_center]{Yunke Ao}\ead{yunke.ao@balgrist.ch},    %
\author[idsc]{Johannes K{\"o}hler}\ead{jkoehle@ethz.ch},   
\author[idsc,eth_cs,eth_ai_center]{Manish Prajapat}\ead{manishp@ai.ethz.ch}, 
\author[eth_cs,eth_ai_center]{Yarden As}\ead{yardas@ethz.ch}, 
\author[idsc,eth_ai_center]{Melanie Zeilinger}\ead{mzeilinger@ethz.ch},
\author[balgrist,eth_ai_center]{Philipp F{\"u}rnstahl}\ead{Philipp.Fuernstahl@balgrist.ch},
\author[eth_cs,eth_ai_center]{Andreas Krause}\ead{krausea@ethz.ch}%

\address[balgrist]{ROCS, Balgrist University Hospital, University of Zurich, Switzerland}
\address[idsc]{Institute for Dynamic Systems and Control, ETH Zurich, Zurich, Switzerland}%
\address[eth_cs]{Department of Computer Science, ETH Zurich, Zurich, Switzerland}
\address[eth_ai_center]{ETH AI Center, ETH Zurich, Zurich, Switzerland}


\begin{abstract}

We propose a stochastic Model Predictive Control (MPC) framework that ensures closed-loop chance constraint satisfaction for linear systems with general \emph{sub-Gaussian} process and measurement noise.
By considering sub-Gaussian noise, we can provide guarantees for a large class of distributions, including time-varying distributions. 
Specifically, we first provide a new characterization of sub-Gaussian random vectors using \emph{matrix variance proxies}, which can more accurately represent the predicted state distribution. 
We then derive tail bounds under linear propagation for the new characterization, enabling tractable computation of probabilistic reachable sets of linear systems.
Lastly, we utilize these probabilistic reachable sets to formulate a stochastic MPC scheme that provides closed-loop guarantees for general sub-Gaussian noise.
We further demonstrate our approach in simulations, including a challenging task of surgical planning from image observations.

\end{abstract}

\end{frontmatter}

\begin{abstract}

We propose a stochastic Model Predictive Control (MPC) framework that ensures closed-loop chance constraint satisfaction for linear systems with general \emph{sub-Gaussian} process and measurement noise.
By considering sub-Gaussian noise, we can provide guarantees for a large class of distributions, including time-varying distributions. 
Specifically, we first provide a new characterization of sub-Gaussian random vectors using \emph{matrix variance proxies}, which can more accurately represent the predicted state distribution. 
We then derive tail bounds under linear propagation for the new characterization, enabling tractable computation of probabilistic reachable sets of linear systems.
Lastly, we utilize these probabilistic reachable sets to formulate a stochastic MPC scheme that provides closed-loop guarantees for general sub-Gaussian noise.
We further demonstrate our approach in simulations, including a challenging task of surgical planning from image observations.

\end{abstract}

\section{Introduction}

Many real-world control systems operate in safety-critical environments.
As such, these systems must maintain safety at all times, \emph{even in light of stochasticity or model ambiguity}.
Model Predictive Control (MPC) is a widely adopted optimization-based control framework, particularly well-suited for addressing challenges related to constraint satisfaction~\cite{rawlings2017model,QIN2003733}. 
Robust and stochastic MPC techniques are commonly used to ensure constraint satisfaction in systems influenced by significant process and measurement noise.

Robust MPC approaches enforce satisfaction of safety guarantees under worst-case scenarios~\cite{mayne2006robust,langson2004robust,richards2005robust}, often leading to overly conservative uncertainty propagation~\cite{bemporad2007robust}.
In contrast, stochastic MPC approaches model noise as random variables with stronger distributional assumptions and enforce constraints with a user-chosen probability, thereby reducing conservatism~\cite{hewing2020recursively,muntwiler2023lqg,farina2016stochastic}.
This work seeks to balance the need for reduced conservatism with weaker assumptions on the underlying noise distribution, by generalizing the existing stochastic MPC methods to sub-Gaussian noise.

Stochastic MPC has been widely studied \cite{farina2016stochastic,lindemann2024formal,Prandini2012scenario,hewing2019scenario}, including theoretical results for closed-loop chance constraint satisfaction~\cite{hewing2020recursively,muntwiler2023lqg,kohler2024predictive}.
A common challenge in these frameworks is the computation of probabilistic reachable sets (PRS), i.e., sets containing future states with a high probability.
Methods proposed by Hewing et al. and Muntwiler et al. 
leverage Gaussian distribution of the noise to derive PRS in closed-form~\cite{hewing2020recursively,muntwiler2023lqg}.
However, the noise in real-world applications is often not Gaussian distributed.
\revision{Sampling-based techniques (conformal prediction or scenario approach) based on independent and identically distributed (i.i.d.) are leveraged in~\cite{lindemann2024formal,Prandini2012scenario,mammarella2022chance}}.
Nevertheless, sampling-based methods can be computationally expensive for long-horizon problems and the i.i.d. assumption may be too restrictive in many applications. 
The Gaussian noise assumption can be relaxed using distributional robustness (DR) approaches, which can provide guarantees for families of distributions  \cite{mark2021data,aolaritei2023wasserstein,li2023distributionally,li2024distributionally}.
For instance, simple computations of PRS can be derived for general distributions using only the covariance, though the resulting sets tend to be conservative~\cite{li2023distributionally,hewing2020recursively,farina2015approach}.
Aolaritei et al. recently incorporated samples and the Wasserstein distance to compute PRS, but the method still relies on i.i.d. noise assumptions~\cite{aolaritei2023wasserstein}.

Despite the advancements of existing works, limited work addresses closed-loop guarantees for 
MPC under \emph{non-Gaussian} and \revision{non-identical} noise \revision{distributions}. 
This challenge is particularly relevant for vision-based control, where states or intermediate observations are estimated from 
images and subsequently used to ensure safe control ~\cite{chou2022safe,li2015vision,drews2019vision}.
In such cases, the estimation error is in general non-Gaussian and \revision{heteroscedastic (non-identical due to correlation with the state)}, see \Cref{rmk:nonlinear_obs} later.

To address this challenge, we draw on the concept of \emph{light-tailed distributions}, widely used in machine learning and high-dimensional statistics~\cite{vershyninHighdimensionalProbabilityIntroduction2018,chowdhury2017kernelized}, as a suitable characterization of such noise distributions.
In particular, \emph{sub-Gaussian} distributions encompass a broad class of light-tailed distributions (e.g., Gaussian) and all bounded distributions (e.g., the uniform distribution)~\cite{vershyninHighdimensionalProbabilityIntroduction2018}.
Furthermore, note that sub-Gaussianity does not require that the noise is identically distributed.

\begin{figure}[t]
\centering
\includegraphics[width=0.38\textwidth]{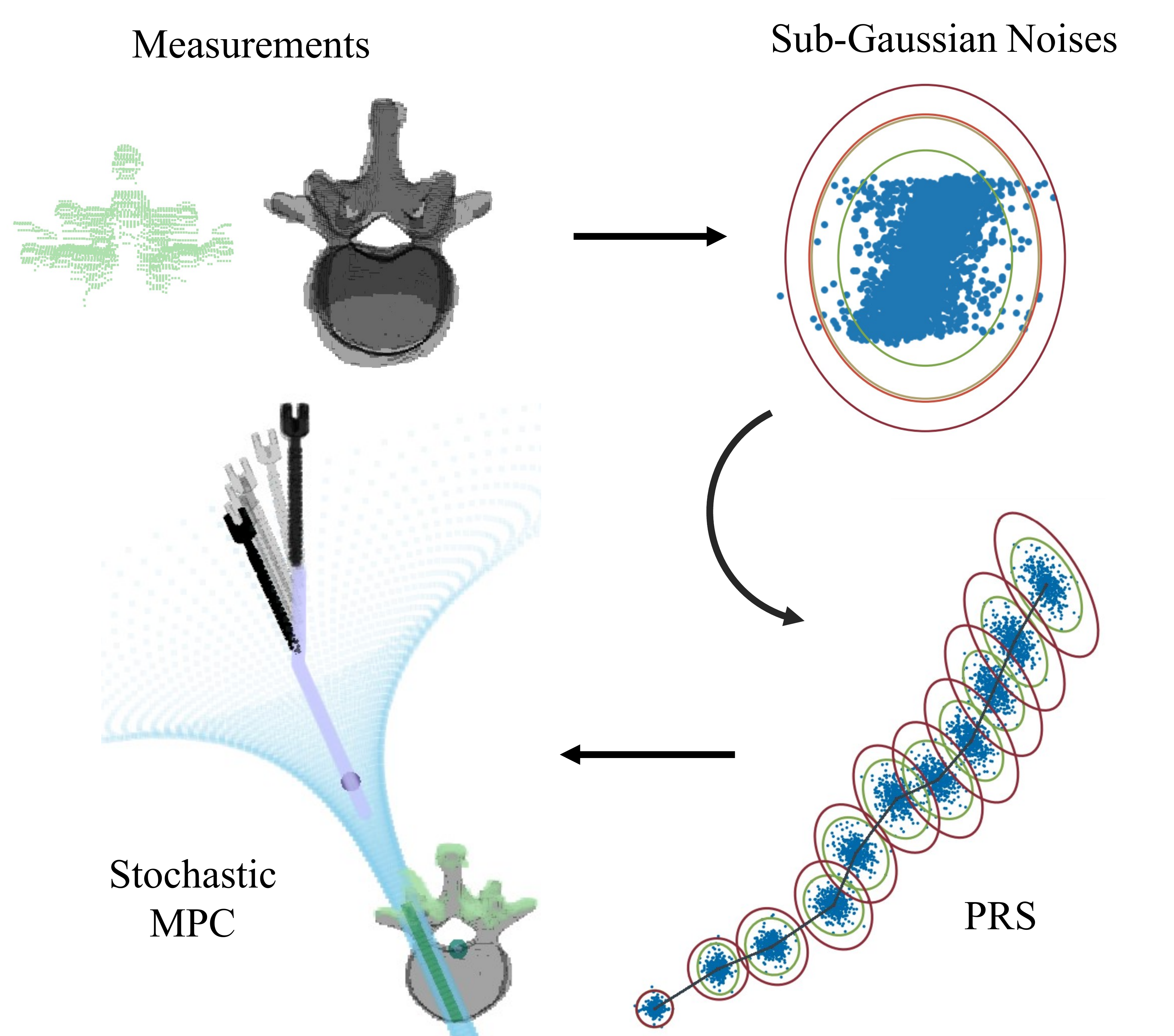}
\caption{
Overview of the proposed stochastic MPC for sub-Gaussian noise at the example of the surgical planning (\Cref{sec:env}). 
We first obtain (high-dimensional) measurements, \revision{estimate the vector state} and compute a sub-Gaussian characterization of the \revision{state estimation error}. %
Then, we provide a simple method to propagate uncertainty and compute probabilistic reachable sets (PRS, Section \labelcref{sec:prob_guar,sec:prop}).
The resulting probabilistic reachable sets of states are utilized in the stochastic MPC to provide probabilistic safety guarantees (Section \labelcref{sec:IF-SMPC-KF}).
}
\label{fig:teaser}
\end{figure}

In this work, we introduce a stochastic MPC framework for sub-Gaussian noise, see~\Cref{fig:teaser} for an overview of the proposed approach.
In particular, we extend the stochastic MPC framework~\cite{muntwiler2023lqg} from Gaussian noise to handle general sub-Gaussian noise.
We show that the resulting closed-loop system satisfies the chance constraints and provides a suitable bound on the asymptotic average performance.
These results are enabled through our technical contributions: 
\begin{enumerate}[(i)]
    \item New characterization of multivariate sub-Gaussian noise using \emph{matrix} variance proxies;
    \item Linear propagation rules for the proposed matrix variance proxies;
    \item Probabilistic reachable sets and moment bounds for the proposed sub-Gaussian characterization.
\end{enumerate} 
Through numerical simulations, we demonstrate the advantages of our approach over existing stochastic, robust, and DR methods.
\looseness=-1

\textit{Notation:}
Let $\|x\|_V$ denote $\sqrt{x^\top Vx}$ for $x\in \mathbb{R}^n$ and $V \in \mathbb{R}^{n\times n}$. 
Let $x_{0:t}$ be $x_0,x_1,x_2,...,x_t$.
\revision{$\|V\|_2$ denotes the matrix norm of $V\in \mathbb{R}^{n\times m} $ induced by the vector 2-norm.}
We use $I$ to represent the identity matrix.
Let $\lambda_{\mathrm{max}}(A)$ denote the maximum eigenvalue of the symmetric matrix $A$.
\revision{We use $\Omega$ to represent the sample space, i.e., the set of possible random noise realizations.}
We denote the expectation by $\E$. 
$\mathcal{N}(\mu, \Sigma)$ denotes the Gaussian distribution with mean $\mu$ and covariance matrix $\Sigma$.
We use $\mathbb{P}_x$ and $\mathbb{P}_{x|y}$ to denote the distribution of $x$ and $x$ given $y$ respectively, i.e. $x\sim \mathbb{P}_x$ and $x\sim \mathbb{P}_{x|y}|y$.
$\mathrm{Pr}\{E\}$ denotes the probability of an event $E$.
Let $\mathbb{N}$ denote the natural number set.
We denote the Minkowski sum by $\oplus$.
We use $\mathcal{K}_\infty$ to denote the set of continuous functions $\alpha: \mathbb{R}_{\geq0} \rightarrow \mathbb{R}_{\geq0}$ which are strictly increasing, unbounded, and satisfy $\alpha({0})=0$.

\section{Problem statement}

\label{sec:prob_set}
We consider the following linear time-invariant system:
\begin{subequations}
\label{eq:sys}
    \begin{align}
    x_{t+1} &=  Ax_t + Bu_t + w_t, \label{eq:dynamics}\\ 
    y_t &= Cx_t + \epsilon_t, \label{eq:obs}
\end{align}
\end{subequations}
where $t\in\mathbb{N}$ is the time step, $x_t\in \mathbb{R}^{n_x}$ is the state of the system, $y_t\in \mathbb{R}^{n_y}$ is the measurement, $u_t\in \mathbb{R}^{n_u}$ is the control input, $w_t \in \mathbb{R}^{n_x}$ are process noise and $\epsilon_t \in \mathbb{R}^{n_y}$ are measurement noise.
The pair $(A,B)$ is stabilizable and $(A,C)$ is detectable.
The states and inputs are subject to chance constraints:
\begin{align}
    &\mathrm{Pr}\{x_t\in\mathcal{X},u_t\in\mathcal{U}\} \geq 1 - \delta,\quad\forall t\in\mathbb{N}, \label{eq:cons}
\end{align}
where $\mathcal{X}$ and $\mathcal{U}$ are safety-critical state and input constraint sets, $1-\delta$ represents the user-specified satisfaction probability.

We consider sub-Gaussian noise distributions.
\begin{definition}[$\sigma-$sub-Gaussian \cite{vershyninHighdimensionalProbabilityIntroduction2018}]
\textit{
    A real-valued random variable \revision{$X:\Omega\rightarrow \mathbb{R}$} with \revision{finite} mean $\mu$ and variance proxy $\sigma$ is $\sigma$-sub-Gaussian, if, \revision{the moment generation function of $X$ exists, and} for all $s\in \mathbb{R}$, we have:
    \begin{align}
    \E\left[\exp{\left(s(X - \mu)\right)}\right] \leq \exp{\left(\frac{\sigma^2s^2}{2}\right)}.
        \label{eq:def_classic}
\end{align}
A real-valued random vector \revision{$X:\Omega\rightarrow \mathbb{R}^n$} is $\sigma$-sub-Gaussian, if the scalar $\lambda^\top X$ is $\sigma$-sub-Gaussian for all $\|\lambda\|=1$.
\label{def:classic}}
\end{definition}
\revision{The left-hand side of~\Cref{def:classic} is the moment generating function of the centered random variable}.
\revision{The right-hand side enforces a light-tailed distribution with (squared) exponential decay~\cite{vershyninHighdimensionalProbabilityIntroduction2018}}.
We denote $\mathbb{P}\in \mathcal{SG}(\mu, \sigma)$ that a distribution $\mathbb{P}$ is sub-Gaussian with mean $\mu$ and variance proxy $\sigma$.
$\mathcal{SG}(\mu, \sigma)$ can characterize a whole class of distributions, such as Gaussian, uniform, and all bounded distributions.
We assume that the initial state, measurement and process noise are (conditionally) sub-Gaussian with known variance proxies.
\begin{assumption}\label{assump:subGaussian}
\textit{For $x_0, w_{0:t}, \epsilon_{0:t}$ in system \labelcref{eq:sys}, we have:
\begin{subequations}
    \begin{align}
    &\mathbb{P}_{x_0}\in \mathcal{SG}(\mu_0, \sigma_0), \\
    &\mathbb{P}_{w_t|x_0,w_{0:t-1},\epsilon_{0:t-1}, u_{0:t-1}} \in \mathcal{SG}(0, \sigma_w),\quad \forall t\in\mathbb{N},\\
&\mathbb{P}_{\epsilon_t|x_0,w_{0:t-1},\epsilon_{0:t-1}, u_{0:t-1}} \in \mathcal{SG}(0, \sigma_\epsilon),\quad \forall t\in\mathbb{N},
\end{align} \label{eq:ass_1}
\end{subequations}
    where $\mu_0\in \mathbb{R}^{n_x},\sigma_0,\sigma_w,\sigma_\epsilon>0$ are known.
}
\end{assumption}
Note that \Cref{assump:subGaussian} does not restrict the distributions of $\epsilon_t$ and $w_t$
to be \emph{identical over time}, as commonly assumed in stochastic MPC literature.
\revision{However, the distributions are conditional zero mean, which can be naturally satisfied in case they are independent.}
Here $\sigma_\epsilon$ and $\sigma_w$ are %
\revision{maximum} sub-Gaussian variance proxies of measurement and process noise \revision{distributions, which may be non-identical over time.
In practice, $\sigma_\epsilon$ and $\sigma_w$ can be estimated from 
samples (cf.~\Cref{sec:env}).
} %

\begin{remark}
\label{rmk:nonlinear_obs}
    The consideration of general sub-Gaussian noise (Assumption~\ref{assump:subGaussian}) allows for %
    \revision{addressing} nonlinear observations from images or point clouds, see also the example in \Cref{sec:env}. 
    Specifically, suppose we have a non-linear observation $I_t = o(x_t) + \eta_t$ with i.i.d. noise $\eta_t$. 
    Typically, we use a model-based algorithm or an offline learned inverse mapping, e.g, through neural networks \cite{chou2022safe},  of the form
    \begin{align*}
        r(I_t) =r(o(x_t) + \eta_t) = Cx_t + \underbrace{r(o(x_t) + \eta_t) - Cx_t }_{=:\epsilon(x_t, \eta_t)}.
    \end{align*} 
    \revision{This yields a linear observation model with noise $\epsilon(x_t, \eta_t)$, which has the same form as~(\ref{eq:obs}).}
    \revision{The resulting noise distribution $\mathbb{P}_{\epsilon(x_t, \eta_t)|x_t}$ is non-identical over different $x_t$}.
    However, if $\epsilon(x_t, \eta_t)$ is bounded for all $x_t$ and zero-mean \revision{(or the means are known and subtracted from the system)}, a common sub-Gaussian variance proxy $\sigma_\epsilon$ exists, such that $\mathbb{P}_{\epsilon(x_t, \eta_t)|x_t}\in \mathcal{SG}(0, \sigma_\epsilon)$ for all $x_t$, i.e., Assumption~\ref{assump:subGaussian} holds.  
\end{remark}

Overall, we consider the following stochastic optimal control problem:
\begin{subequations} \label{eq:soc_problem}
\begin{align}
    & \inf_{\pi_{0:\infty}} \revision{\limsup_{T\rightarrow\infty}\frac{1}{T}}\sum_{t=0}^{\revision{T}}\ell(x_t, u_t) \label{eq:cost}\\
    \text{s.t. } &  u_t = \pi_t(y_{0:t}, u_{0:t-1}),
    ~\labelcref{eq:sys},~\labelcref{eq:cons},~\labelcref{eq:ass_1},~\forall t\in \mathbb{N},
\end{align}
\end{subequations}
where $\ell$ is the stage cost and $\pi_t$ are dynamic output-feedbacks. 
In this paper, we present a tractable approach to solving Problem~\eqref{eq:soc_problem}.

\section{Method}
\label{sec:thr}
In what follows, we develop our theory and analysis for solving Problem \labelcref{eq:soc_problem}.
We first provide a new definition of sub-Gaussian random variables using a \emph{matrix variance proxy}.
Further, in \Cref{sec:prop},  we introduce linear propagation rules using such a matrix variance proxy.
We then derive confidence bounds and moment bounds of the proposed new sub-Gaussian characterization in \Cref{sec:prob_guar}.
These results are finally 
utilized to extend the state-of-the-art stochastic output-feedback MPC framework for Gaussian noise \cite{muntwiler2023lqg} to solve the Problem \labelcref{eq:soc_problem} (\Cref{sec:IF-SMPC-KF}).

\subsection{Sub-Gaussian with matrix variance proxy}
\label{sec:sub_gau_def}

\Cref{def:classic} characterizes sub-Gaussian random vectors with a scalar variance proxy $\sigma$.
However, in linear systems \labelcref{eq:sys}, stochastic variances of states often develop correlations or scale differences across dimensions as they propagate through the dynamics. 
Consequently, relying on scalar variance proxies tends to overestimate uncertainty for state dimensions with smaller variance. 
To address this, we introduce a definition of sub-Gaussian random vectors using a matrix variance proxy.
\begin{definition}[Sub-Gaussian with matrix (co-)variance proxy]\label{def:sub_gau_mean}
    A real-valued random vector $\revision{X:\Omega\rightarrow\mathbb{R}^n}$ with \revision{finite} mean $\E[X] = \mu$ is called sub-Gaussian with a variance proxy $\Sigma\succeq 0$, i.e., $X\sim \mathcal{SG}(\mu,\Sigma)$, if, \revision{the moment generation function of $X$ exists and}$~\forall~\lambda \in \mathbb{R}^n,$
    \begin{align}
    \E\left[\exp{\left(\lambda^\top(X - \mu)\right)}\right] \leq \exp{\left(\frac{\|\lambda\|^2_\Sigma}{2}\right)}.\label{eq:sub_gau_def}
    \end{align}
\end{definition}
\revision{Similar to \Cref{def:classic}, the left-hand side is the moment generating function of the centered vector-valued random variable.
\Cref{def:sub_gau_mean} characterizes high-dimensional light-tailed distributions whose decay rates can vary across dimensions, characterized by the matrix~$\Sigma$.}
Next, we show that \Cref{def:sub_gau_mean} generalizes the standard definition, i.e., \Cref{def:classic} is a special case of \Cref{def:sub_gau_mean}.
\begin{lemma}
\textit{
    Every $\sigma$-sub-Gaussian random vector satisfying \Cref{def:classic} also has a finite matrix variance proxy $\Sigma=\sigma^2 I$ with \Cref{def:sub_gau_mean}, and vice versa, i.e., every sub-Gaussian random vector having a matrix variance proxy $\Sigma \succ 0$ with \Cref{def:sub_gau_mean} is \revision{$\sigma=\sqrt{\|\Sigma\|_2}$}-sub-Gaussian with \Cref{def:classic}.}
    \label{lemma:equ}
\end{lemma}

The proof of this lemma is detailed in \Cref{sec:proof_equ}.
Consequently, as \emph{all distributions with bounded support} are sub-Gaussian noises under \Cref{def:classic} \cite{vershyninHighdimensionalProbabilityIntroduction2018}, they are also sub-Gaussian with a matrix variance proxy.
We note that the multivariate sub-Gaussian stable distribution \cite{nolan2013multivariate,swihart2022multivariate} also uses positive definite matrices to characterize light-tailed distributions. 
However, this characterization can only capture elliptically contoured distributions \cite{cambanis1981theory}, i.e.,  distributions whose probability mass contours are elliptically shaped, while \Cref{def:sub_gau_mean} has no such limitations.
Moreover, contrary to \Cref{def:sub_gau_mean}, this characterization does not contain the scalar sub-Gaussian definition as a special case.

\subsection{Uncertainty propagation with linear systems}
\label{sec:prop}

In System \labelcref{eq:sys}, states are propagated under linear transformation and addition.
Here we show that sub-Gaussian distributions are closed under these operations and the resulting propagation of matrix variance proxy is also straightforward.

\begin{theorem}[Propagation of matrix variance proxy]\label{thr:propagation} 
\textit{
    Consider $X\sim \mathcal{SG}(\mu,\Sigma)$ (\Cref{def:sub_gau_mean}) with $\mu\in\mathbb{R}^n$ and $\Sigma\succeq 0\in\mathbb{R}^{n\times n}$.
    \begin{enumerate}[a.]
        \item For any matrix $A\in \mathbb{R}^{m\times n}$, $AX\sim \mathcal{SG}(A\mu, A\Sigma A^\top)$.
        \item If $~\mathbb{P}_{Y|X} \in \mathcal{SG}(\mu',\Sigma')$, then \\$\mathbb{P}_{X + Y}\in\mathcal{SG}(\mu + \mu',\Sigma + \Sigma')$.
    \end{enumerate}
    }
\end{theorem}

\begin{proof}

From Definition~\ref{def:sub_gau_mean}, we have for $a$:
     \begin{align*}
        &\E\left[\exp\left({\lambda^\top A(X - \mu)}\right)\right] \leq \exp{\left(\frac{\|A^\top\lambda\|_\Sigma^2}{2}\right)} \\
        = &\exp{\left(\frac{\|\lambda\|_{A\Sigma A^\top}^2}{2}\right)}.
    \end{align*}
    For $b$, it can be shown by:
    \begin{align*}
        &~~~~ \E\left[\exp{\left[\lambda^\top((X  - \mu) + (Y - \mu')) -\revision{\dfrac{\|\lambda\|^2_\Sigma + \|\lambda\|^2_{\Sigma'}}{2}}\right]}\right]  \\
        &= \E_X\left[\exp{\left(\lambda^\top(X - \mu)-\dfrac{\|\lambda\|^2_{\Sigma}}{2}\right)}\right. \\
        & ~~~\left.\E_{Y|X}\left[\exp{\left(\lambda^\top(Y - \mu')- \dfrac{\|\lambda\|^2_{\Sigma'}}{2}\right)}\right]\right] \\
        &\stackrel{\revision{\labelcref{eq:sub_gau_def}}}{\leq} 
        \begin{aligned}[t]\E_X\left[\exp\left({\lambda^\top(X - \mu)-\dfrac{\|\lambda\|^2_{\Sigma}}{2}}\right) \cdot 1\right] 
        \stackrel{\revision{\labelcref{eq:sub_gau_def}}}{\leq} 1.\qedhere
        \end{aligned}
    \end{align*}
\end{proof}

Theorem~\ref{thr:propagation} indicates that the propagation rule of matrix variance proxy is similar to the propagation of covariance matrices, enabling simple uncertainty propagation with linear systems.
Propagation of sub-Gaussian noise under linear dynamics has also been studied for system identification~\cite{sarker2023accurate}, 
however, using a scalar variance proxy and without derivations of probabilistic reachable sets.

\subsection{Confidence and moment bounds}
\label{sec:prob_guar}
In stochastic MPC, one key step for guaranteeing safety is computing probabilistic reachable sets (PRS), 
i.e., establishing confidence bounds $\mathcal{E}_t^x$ for the state distributions with $\mathrm{Pr}\{x_t\in \mathcal{E}_t^x\} \geq 1 - \delta$.
With \Cref{thr:propagation}, we can predict matrix variance proxies of state distributions in system \labelcref{eq:sys}.
To compute PRS, we additionally need to derive confidence bounds using these obtained matrix variance proxies.
  Next, we present two confidence bounds for sub-Gaussian distributions.
\begin{lemma}[Half-space bound]\label{lemma:half_space}
    \textit{
    If $X\sim\mathcal{SG}(\mu,\Sigma)$, then for any $h\in \mathbb{R}^n$, 
    $\mathrm{Pr}\{ X\in \mathcal{E}^h\}\geq 1-\delta$
    with the half-space confidence bound:
    \begin{align*}
        \mathcal{E}^{\mathrm{h}}(\mu, \Sigma, \delta, h):=\left\{X\,\,|\,\,h^\top(X - \mu) \leq \|h\|_{\Sigma} \sqrt{2\ln\frac{1}{\delta}}\right\}.
    \end{align*}}
\end{lemma}
\begin{proof}
By Chernoff inequality, for any $s>0$ and $\tau>0$,
\begin{align*}
    &\mathrm{Pr}\left\{h^\top(X - \mu) \geq \tau\right\} \\
    = &\,\mathrm{Pr}\left\{\exp\left({sh^\top(X - \mu)}\right) \geq \exp\left({s\tau}\right)\right\} \\
    \leq &\,\E\left[\exp\left({sh^\top(X - \mu)-s\tau}\right)\right]\stackrel{\revision{\labelcref{eq:sub_gau_def}}}{\leq} \exp\left({\frac{s^2\|h\|^2_{\Sigma}}{2} - s\tau}\right).
\end{align*}
Assigning $s= \dfrac{\tau}{\|h\|^2_{\Sigma}}$ gives:
\begin{align*}
    \mathrm{Pr}\{h^\top(X - \mu) \geq \tau\} \leq \exp\left(-\frac{\tau^2}{2\|h\|^2_{\Sigma}}\right).
\end{align*}
Then solving $\tau$ from $\exp\left(-\frac{\tau^2}{2\|h\|^2_{\Sigma}}\right)=\delta$ yields the confidence bound for $1-\delta$.
\end{proof}
Note that this bound recovers the known confidence bound for scalar variance proxy \cite{vershyninHighdimensionalProbabilityIntroduction2018} as a special case.

Given Lemma~\labelcref{lemma:half_space}, we could also construct polytope confidence sets as an intersection of individual half-space constraints using Boole's inequality~\cite{paulson2020stochastic}. 
To leverage the correlation between different dimensions, we also introduce elliptical confidence bounds using the variance proxy $\Sigma$ more directly:
\begin{theorem}[Elliptical bound]\label{thr:elliptical}
\textit{
Consider $X\sim \mathcal{SG}(\mu,\Sigma)$ with $\mu\in\mathbb{R}^n$ and $\Sigma\succ 0\in\mathbb{R}^{n\times n}$, then we have for all $\tau>\sqrt{n}$:
    \begin{align}
        \mathrm{Pr}\{\|X - \mu\|_{\Sigma^{-1}} \geq \tau\} \leq \left(\frac{e}{n}\right)^{\frac{n}{2}}\tau^n \exp{\left(-\frac{\tau^2}{2}\right)} \label{ineq:tail}
    \end{align}
     Moreover, $\mathrm{Pr}\{ X\in \mathcal{E}^e\}\geq 1-\delta$ with the elliptical confidence bound:
    \begin{align}
    \mathcal{E}^{\mathrm{e}}(\mu, \Sigma, \delta, n):=\left \{X \,|\, \|X - \mu\|^2_{\Sigma^{-1}} \leq n + ng^{-1}(\delta^{-\frac{2}{n}})\right\},\label{eq:e_set}
    \end{align}
    where $g\in \mathcal{K}_\infty,\,g(x)=\dfrac{\exp{x}}{1+x}$. }
\end{theorem}
The proof of this theorem is detailed in \Cref{sec:proof_tail}. 
Moreover, \Cref{thr:elliptical} can also give a cylindrical set with bounds only in a subspace as 
\begin{align}
\mathcal{E}^e=\underbrace{H^{\dagger}\mathcal{E}^e(H\mu, H\Sigma H^\top,\delta,n_c)\oplus \mathrm{Null}(H)}_{=:\mathcal{E}^e_H(H,\mu,\Sigma,\delta,n_c)}, \label{eq:set_subspace}
\end{align} 
where $H\in \mathbb{R}^{n_c\times n},n_c<n$, $H^\dagger$ denotes the pseudo-inverse of $H$, and $\mathrm{Null}(H)$ represents the null space $\{x\in \mathbb{R}^n|Hx=0\}$.
Clearly, this set only has an elliptical boundary in the subspace $\mathrm{span}(H)$ and unrestricted in $\mathrm{Null}(H)$. 

Similar to the Gaussian case, our sub-Gaussian confidence bound grows logarithmically w.r.t. $\delta^{-1}$:
\begin{corollary}\label{lemma:complexity}
    For all $\delta \in (0,1)$, $n\geq 1$, \revision{$x\in \mathcal{E}^{\mathrm{e}}(\mu, \Sigma, \delta, n)$} with $\mathcal{E}^e$ in \Cref{thr:elliptical}, \revision{it holds that} 
    \begin{align*}
        \revision{\|x-\mu\|_{\Sigma^{-1}}^2\leq (1 + \ln4)n + 4\ln \delta^{-1}.}
\end{align*}
\end{corollary}
The proof of this corollary is detailed in \Cref{sec:proof_complexity}.
Compared to the bound for distributions only with variance available in \cite{hewing2020recursively} which is $\mathcal{O}(n\delta^{-1})$, our bound is $\mathcal{O}\left(n + \ln\delta^{-1}\right)$ and thus less conservative 
for small $\delta$.
We also provide bounds for the moments of the norm of sub-Gaussian random vectors similar to~\cite[Proposition 2.5.2 (ii)]{vershyninHighdimensionalProbabilityIntroduction2018}:
\begin{lemma}[Bounds of moments]\label{lemma:mom_bound}
    Consider $X\sim \mathcal{SG}(\mu,\Sigma)$ with $\mu\in\mathbb{R}^n$ and $\Sigma\succ 0\in\mathbb{R}^{n\times n}$. For any $p\geq 1$, it holds that
    \begin{align*}
        \E\left[\|X-\mu\|_{\Sigma^{-1}}^p\right] \leq \underbrace{p2^{\frac{p - 1}{2}}\left(\frac{2e}{n}\right)^{\frac{n}{2}}\Gamma\left(\frac{n+p+1}{2}\right)}_{=:\mathcal{B}(p,n)}
    \end{align*}
    where $\Gamma$ is the Gamma function.
\end{lemma}
\begin{proof}
    Similar to~\cite[Proposition 2.5.2 (ii)]{vershyninHighdimensionalProbabilityIntroduction2018}, we have:
    \begin{align*}
        &\E\left[\|X-\mu\|_{\Sigma^{-1}}^p\right] = \int_0^\infty \mathrm{Pr}\{\|X-\mu\|^p_{\Sigma^{-1}} \geq u\} du \\
        &\stackrel{u=t^p}{=} \int_0^\infty \mathrm{Pr}\{\|X-\mu\|_{\Sigma^{-1}} \geq t\} pt^{p-1}dt \\
        & \stackrel{Equ.~\labelcref{ineq:tail}}{\leq}\left(\frac{e}{n}\right)^{\frac{n}{2}}\int_0^\infty p t^{n+p-1}\exp\left({-\frac{t^2}{2}}\right) dt \\
        & \stackrel{\tau=\frac{t^2}{2}}{=}p2^{\frac{p-1}{2}}\left(\frac{2e}{n}\right)^{\frac{n}{2}}\underbrace{\int_0^\infty  \tau^{\frac{n+p-1}{2}}e^{-\tau} dt }_{\Gamma\left(\frac{n+p+1}{2}\right)}\qedhere.
    \end{align*}
\end{proof}
\Cref{lemma:mom_bound} will be useful for analyzing \revision{the stability of} MPC later.
Both \Cref{lemma:half_space} and \Cref{thr:elliptical} yield probabilistic reachable sets that can be leveraged in the stochastic MPC scheme.
\Cref{lemma:half_space} is ideal if~\labelcref{eq:cons} is a single half-space constraint and it can also be applied for polytope chance constraints.
\Cref{thr:elliptical} is capable of handling general constraints.

\revision{\section{Sub-Gaussian stochastic MPC}}
\label{sec:IF-SMPC-KF}
In this section, we address Problem \labelcref{eq:soc_problem} by extending the indirect output-feedback MPC framework \cite{muntwiler2023lqg} from Gaussian to sub-Gaussian noise.
\revision{Indirect output-feedback MPC~\cite{muntwiler2023lqg} ensures satisfaction of chance constraints by quantifying the error between the the true trajectory and a nominal trajectory, and solving a nominal MPC problem with tightened constraints.}
\subsection{State estimator and tracking controller}
As in~\cite{muntwiler2023lqg}, \revision{we use a nominal state $z$ and implement a dynamic output-feedback to keep the estimated state $\hat{x}$ close to the real state $x$ and the nominal state $z$ using:}
\begin{subequations}\label{eq:track_est}
    \begin{align}
        z_{t+1} &= Az_t + Bv_t \label{eq:track}\\
        \hat{x}_{t+1} &= A\hat{x}_{t} + Bu_t + L\left(y_{t+1} - C(A\hat{x}_{t} + Bu_t)\right)\label{eq:est}\\
        u_t &= K(\hat{x}_t - z_t) + v_t
    \end{align}
    \label{eq:close_loop}
\end{subequations}
where $z_0=\mu_0$, and $v_t$ is the nominal input. 
The observer gain $L$ and the feedback $K$ are designed offline, e.g., using linear–quadratic–Gaussian \revision{control law~\cite{kwakernaak1972linear}}.

\subsection{Uncertainty propagation}
The error $e_t:=[\hat{x}_t - x_t;x_t - z_t]\in \mathbb{R}^{2n_x}$ consisting of estimation error and tracking error satisfies
\begin{align}
    \label{eq:err_dyn}
    e_{t+1} &= A^e e_t + B^e_1 w_t + B^e_2 \epsilon_t,\\
    A^e & :=\begin{bmatrix}
        A-LCA & 0 \\
        -BK & A + BK
    \end{bmatrix}, \nonumber\\ 
    B^e_1  &:= \begin{bmatrix}
        I - LC \\
        I  
    \end{bmatrix},\, 
    B^e_2 := \begin{bmatrix}
        -L \\
        0  
    \end{bmatrix},\nonumber
\end{align}
with $A^e$ Schur-stable by designing $K,L$ properly.
By denoting the matrix variance proxy of $e_t$ as $\Sigma_t$, we can propagate it through time based on \Cref{thr:propagation}:
\begin{align}
    \Sigma_{t+1} = A^e\Sigma_t{A^e}^\top + \sigma_w^2B^e_1{B^e_1}^\top + \sigma_\epsilon^2 B^e_2{B^e_2}^\top,
    \label{eq:prop}
\end{align} 
where $\Sigma_0=\sigma_0^2I$.

\subsection{Probabilistic reachable sets}
\revision{In this section, we derive the PRS of System~(\ref{eq:close_loop}).
\begin{lemma}
    Consider the closed-loop system~(\ref{eq:close_loop}),~(\ref{eq:dynamics}). Let~\Cref{assump:subGaussian} holds. Define $K^e := \left[\begin{array}{cc}
 0 &  I \\
 K &  K
\end{array}\right]$, and $\bar{\xi}_t:=(z_t, v_t)$.
For any nominal input $v_t$ that depends causally on $w_{0:t-1},\epsilon_{0:t-1}$, the resulting
 state-input trajectory $\xi_t:=(x_t,u_t)$ satisfies $\mathrm{Pr}\{\xi_t\in \bar{\xi}_t\oplus\mathcal{E}_t\}\geq 1 - \delta$ for any of the following sets $\mathcal{E}_t$ defined from \Cref{lemma:half_space},~\Cref{thr:elliptical} and \Cref{eq:set_subspace}:
\begin{align*}
        \mathcal{E}_t &:=  \mathcal{E}^\mathrm{h}(0, K^e\Sigma_t{K^e}^\top, \delta, h), \\
        \mathcal{E}_t &:= \mathcal{E}^\mathrm{e}(0, K^e\Sigma_t{K^e}^\top, \delta, n_x+n_u), \\
        \mathcal{E}_t &:= \mathcal{E}^\mathrm{e}_\mathrm{H}(H, 0, K^e\Sigma_t{K^e}^\top, \delta, n_c),
    \end{align*}
    with $H\in \mathbb{R}^{n_c\times n_x+n_u}, n_c< n_x+u_u$.
    \label{lemma:PRS}
\end{lemma}
\begin{proof}
   Since $\xi_t - \bar{\xi}_t = K^e e_t$, $\xi_t - \bar{\xi}_t$ is sub-Gaussian with variance proxy $K^e\Sigma_t{K^e}^\top$.
    The result then follows
 from \Cref{lemma:half_space},~\Cref{thr:elliptical}, and~\Cref{eq:set_subspace}. 
\end{proof}
Then, the tightened constraints
    \begin{align*}
        (z_t,v_t) \in (\mathcal{X}\times\mathcal{U}) \ominus \mathcal{E}_t
   \end{align*}
 ensure the satisfaction of the chance constraints~(\ref{eq:cons}).}

\subsection{Model predictive control problem formulation}
Following \cite{muntwiler2023lqg}, the MPC problem at each time step $t$ with horizon $H$ is
\begin{subequations}\label{prob:mpc}
    \begin{align}
        \min_{v_{0:H-1|t}} & \ell_f(\bar{x}_{H|t}) + \sum_{i=0}^{H-1} \ell\left(\bar{x}_{i|t}, v_{i|t} + K(\bar{x}_{i|t} - z_{i|t})\right) \label{eq:cost}\\ 
        \mathrm{s.t.}~ & \forall \,i\in\{0,...,H-1\}:\\
        & z_{i+1|t} = Az_{i|t} + Bv_{i|t}, \label{eq:z_prop}\\
        & \bar{x}_{i+1|t} = A\bar{x}_{i|t} + BK(\bar{x}_{i|t} - z_{i|t}) + Bv_{i|t}, \label{eq:x_prop}\\
        &(z_{i|t},v_{i|t})\in (\mathcal{X}\times\mathcal{U}) \ominus \mathcal{E}_{t+i}, \label{eq:mpc_cons}\\
        & z_{H|t}\in \mathcal{Z}_f, \\
        &\bar{x}_{0|t}=\hat{x}_t,\label{eq:x_init}\\
        &z_{0|t}=z_t,\label{eq:z_init}
    \end{align}
\end{subequations}
where $z_{i|t},\bar{x}_{i|t}$ denote the nominal and certainty equivalent prediction of the states predicted $i$ steps in the future.
The optimal nominal inputs at time step $t$ are denoted by $v^*_{0:H|t}$.
Problem \labelcref{prob:mpc} minimize the cost of the prediction conditioned on the estimated state, while constraints are enforced through a nominal initialization with the offline computed PRS $\mathcal{E}_{t:t+H-1}$.
It is a convex quadratic program if $\ell,\ell_f$ are quadratic functions and the constraints are polytopic.
We design the terminal set $\mathcal{Z}_f$ and terminal cost $\ell_f$ such that they satisfy the terminal invariance property:
\begin{assumption}[Terminal set and cost \cite{muntwiler2023lqg}]
\label{ass:terminal}
    The terminal set $\mathcal{Z}_f$ and terminal cost $\ell_f$ satisfy for all $z\in \mathcal{Z}_f$ and all $x\in \mathbb{R}^n$: 
    \begin{enumerate}[a.]
        \item (Positive invariance) $(A + BK)z\in \mathcal{Z}_f$;
        \item (Constraints satisfaction) \\ $(z,Kz) \in (\mathcal{X}\times \mathcal{U})\ominus \mathcal{E}_t, t\in\mathbb{N}$,
        \item (Lyapunov) $\ell_f((A+BK)x)\leq \ell_f(x)-\ell(x,Kx)$.
    \end{enumerate}
\end{assumption}
Here $\mathcal{Z}_f$ can be designed as the maximal positively invariant set of $\{z~|~(z,Kz)\in (\mathcal{X}\times \mathcal{U})\ominus \cup_{t=0}^\infty\mathcal{E}_t\}$.

The resulting closed-loop system is given by:
\begin{align}
    v_t = v^*_{0|t},~\labelcref{eq:track_est} \label{eq:policy}
\end{align}
In order to provide closed-loop stability, we also consider the following regularity conditions:
\begin{assumption}[Regularity conditions]\label{ass:quadratic}
    The cost is given by $\ell(x,u)=\|x\|_Q^2 + \|u\|_R^2,~\ell_f(x)=\|x\|_P^2$ with $Q, R, P\succ 0$. 
\end{assumption}
The matrix $P$ can be computed using the LQR.
The closed properties of the controller \labelcref{eq:policy} are summarized in the following theorem:
\begin{theorem}[Closed-loop Properties]\label{thr:mpc}
Let Assumptions \labelcref{assump:subGaussian,ass:terminal} hold and suppose that Problem~\labelcref{prob:mpc} is feasible at $t=0$. 
Then, the Problem \labelcref{prob:mpc} is recursively feasible for all $t\in\mathbb{N}$, and the closed-loop system~\labelcref{eq:sys}, \labelcref{eq:policy} satisfies the chance constraints \labelcref{eq:cons} for all $t\in\mathbb{N}$.
Furthermore, with \Cref{ass:quadratic}, the asymptotic average cost satisfies:
\begin{align*}
    \lim_{T\rightarrow \infty}\frac{1}{T}\sum_{t=0}^{T-1}\E\left[\ell(x_t, u_t)\right] \leq \kappa_w\left(\sigma_w\right) + \kappa_\epsilon\left(\sigma_\epsilon\right),
\end{align*}
where $\kappa_w$ and $\kappa_\epsilon$ are $\mathcal{K}_\infty$ functions.
\end{theorem}
\begin{proof}
\ifautomatica
\else
The proof is detailed in Appendix~\ref{sec:proof_mpc}.
\fi
\end{proof}
In \Cref{thr:mpc}, the closed-loop constraint satisfaction property provides safety guarantees, while the asymptotic average cost bound implies a low average cost if the variance prox of the noise is small.
Compared to \cite{muntwiler2023lqg}, \Cref{thr:mpc} 
additionally address non-identical noises under sub-Gaussian assumptions.

\section{Numerical experiments}
\label{sec:exp}

In this section, we assess the performance of our uncertainty propagation and MPC methods. 
For our experiments, we empirically demonstrate that
\begin{enumerate}
    \item The PRS computed by our method satisfies the user-specified containment probability, including heteroscedastic noise settings.
    \item Our PRS is less conservative than the robust and distributional robust baselines.
    \item Our MPC approach achieves a smaller cost than existing Distributionally Robust (DR) MPC approaches, while providing probabilistic guarantees on constraint satisfaction.
\end{enumerate}
All experiments are conducted with 
the probability threshold $1-\delta=95\%$. 
\revision{The implementation details are available at \url{https://github.com/ToolManChang/Sub_Gaussian_MPC} \ifautomatica\else.\fi}

\subsection{Environments}
\label{sec:env}
We demonstrate the performance of our approach on \revision{two} different test-beds. 

\textbf{Mass-Spring-Damper (MSD)} 
\revision{The MSD system has the form 
\begin{align*}
    x_{t+1} = \begin{bmatrix}
        1 & \Delta t\\
        -\frac{k\Delta t}{m} & 1 - \frac{b\Delta t}{m}
    \end{bmatrix} x_t + \begin{bmatrix}
        0 \\
        \frac{\Delta t}{m}
    \end{bmatrix} u_t + w_t,
\end{align*}
where $x\in\mathbb{R}^2$ is the stacked position and velocity, $m$ is the mass, $b$ is the damping coefficient, $k$ is the spring constant, and $\Delta t$ is the time step of the discretized system.
We choose $\Delta t=0.1$, $m = 2$, $k=1$ and $b=1$. The observation model is simply $y_t = x_t + \epsilon_t$. The noise $w_t, \epsilon_t$ are sampled randomly from Student-t distributions.
We truncate the distributions to ensure the noise remains sub-Gaussian. 
Heteroscedastic noises are introduced by setting the noise to be 5 times greater when $x[0]>0.2$.}
\revision{The target state is defined as $x^* = [0.5, 0.0]^\top$, and our cost $\ell(x_t, u_t)$ is defined as $\|x_t - x^*\|^2$.
The constraint is defined as $x_t[0] \leq 0.5,\forall\,\, t\in\mathbb{N}$, where $[\cdot]$ denotes the index of dimension.
}

\textbf{Surgical Planning (SP)}
\revision{This environment, taken from~\cite{ao2025saferplan}, provides a simplified model for intraoperative pedicle screw placement, a common step for robotic spine surgery.}
\revision{We define the state as the relative pose between the surgical tool and the target pose $x_t:=[p^d_t - p, q^d_t]^\top$, where $p^d_t \in \mathbb{R}^3$ and $q^d_t \in \mathbb{R}^2$ are position and sphere coordinates of the direction (deviated from the goal) of the tool, $p$ is the target position.
The unknown target position $p$ (on a real patient anatomy) is estimated as $\Tilde{p}$ by registration between intraoperative anatomy reconstruction and the preoperative image.
The observation is defined as the per-step estimated state $y:= [p^d_t - \Tilde{p}_t, q^d_t]^\top$.
The resulting dynamics is:
\begin{align*}
    x_{t+1} &= x_t + u_t \Delta t + w_t \\
    y_t &= x_t + \epsilon_t
\end{align*}
where $\Delta t=0.075$, $\epsilon_t$ is the per-step state estimation errors following unknown distributions.
The goal state is $x^*=[0.12, 0.0, 0.0, 0.0, 0.0]^\top$.
The cost function $\ell(x_t, u_t)$ is defined as $\|x_t - x^*\|^2 + 0.001\|u_t\|^2$.
The safety constraint sets for all $0\leq t\leq T$ are described by:
\begin{align}
   \|x_t[1:2]\| &\leq \frac{1}{5}\sqrt{\exp(-2500 x_t[0]^2 - 5)+ 0.0004}, \nonumber\\
   x[0] &\leq 0.12 \label{eq:sp_x_cons}
\end{align}
where a funnel-like narrow feasible region \revision{(Figure~\ref{fig:teaser}, light blue)} is constructed to simplify the safety constraint of the real surgery.
}

Since the analytical sub-Gaussian variance proxies are not available for our test noise \revision{and registration error} distributions, we use 5000 samples to calibrate their variance proxies, akin to \cite{arbel2020strict}. 
Specifically, from \Cref{def:sub_gau_mean}, it holds that:
\revision{\begin{align*}
    \sigma^2 = \max_{\lambda\in \mathbb{R}^n} \frac{2\ln \E[e^{\lambda^\top(s - \mu)}]}{\|\lambda\|^2} \approx \max_{\lambda\in \mathbb{R}^n} \frac{2\ln \frac{1
    }{N}\mathop{\textstyle\sum}\limits_{i=1}^N e^{\lambda^\top(s_i - \hat{\mu})}}{\|\lambda\|^2}
\end{align*}}
where $s_{1:N}$ are data samples and $\hat{\mu}$ is the sample mean.
\revision{For MSD with heterosedastic noise, we calibrate the maximum variance proxy over all noise distributions.}

\begin{figure}
    \centering
    \includegraphics[width=0.4\textwidth]{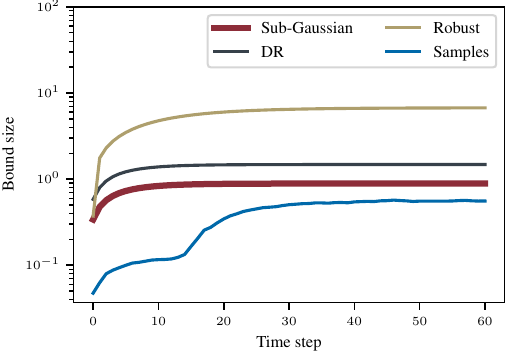}
    \caption{Comparison of $95\%$ confidence bound sizes quantified by different methods in the mass-spring-damper environment with \revision{truncated Student-t} noise.
    \revision{Blue} lines represent quantiles from test samples.
    \revision{All} approaches compute the global maximum confidence bound to address the heteroscedastic noise.
    \revision{The confidence bound from our sub-Gaussian approach is greater than the quantiles but less conservative than robust and distributionally robust approaches. }
    }
    \label{fig:up}
\end{figure}

\begin{figure*}[htbp]
\centering
\begin{subfigure}[t]{0.48\textwidth}
    \includegraphics[width=\textwidth]{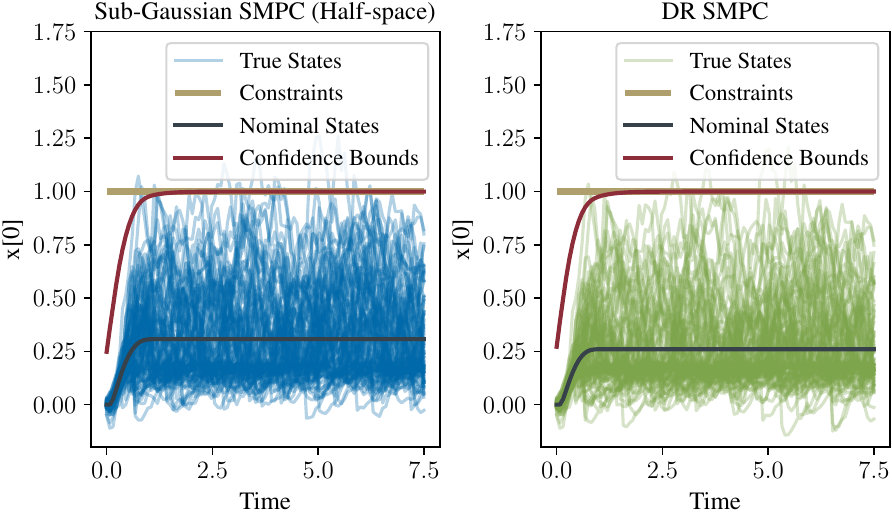}
    \label{fig:mpc_msd}
\end{subfigure}
\begin{subfigure}[t]{0.48\textwidth}
    \includegraphics[width=\textwidth]{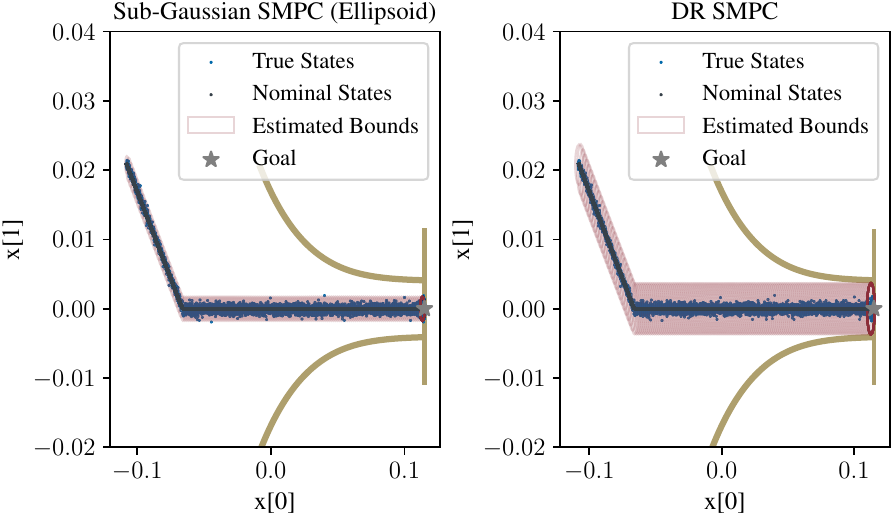}
    \label{fig:mpc_vl_sp}
\end{subfigure}
\caption{
Plans from sub-Gaussian \revision{and DR} MPC approaches in 100 trials from \revision{(left)} MSD \revision{and (right)} SP environments, respectively.
For MSD, the $x$ and $y$ axes are time and the first state, respectively.
For SP, they correspond to the first 2 dimensions of the states.
The confidence levels of displayed examples are set at $95\%$.
The yellow lines represent the boundary constraints.
In all problems, the proposed approach satisfies the safety-critical constraints with the chosen probability $95\%$.
}
\end{figure*}

\subsection{Baselines}
\label{sec:baselines}
\revision{
We consider two baselines for comparison: Robust and DR approaches.
\ifautomatica
\else
\newline\textbf{Robust approach~\cite{mayne2006robust,richards2005robust}}
\label{sec:robust}
The noise terms $\epsilon_t$ and $w_t$ are bounded within sets $\mathcal{E}$ and $\mathcal{W}$ respectively, which are calibrated as the maximum bound from samples.
The uncertainty propagation in \Cref{eq:err_dyn} is handled through set propagation:
    $\mathcal{E}_{t+1} = A\mathcal{E}_t\oplus B_1^e \mathcal{W}\oplus B_2^e \mathcal{E}$,
where $\oplus$ is the Minkowski sum.\\\newline
\textbf{Distributionally Robust (DR) approach~\cite{hewing2020recursively,farina2015approach}} 
Instead of Gaussian distributions, many works consider formulations that treat \emph{all distributions} with the given covariance matrix.
Specifically, with the same covariance matrix propagated as Stochastic-Gaussian approaches, the bounds are obtained with Chebyshev inequality $\mathcal{E}_t = \{e_t \,\,|\,\, e_t^\top{\Sigma_t^{-1}}e_t \leq \frac{n_c}{\delta}\}$.
In this case, the resulting bounds are distribution-agnostic, which comes at the price of increased conservatism.
but at the price of being loose and much larger than the quantile function of $\chi^2(2n_c)$ distribution.
\fi
}

\subsection{Uncertainty propagation}

\label{sec:exp_up}

In this section, we study the performance of our approach compared to the baseline uncertainty propagation methods.
To this end, we use different approaches to predict probabilistic reachable sets $\mathcal{E}_t$ (\Cref{sec:baselines}) conditioned on the same action sequence $u_t$ under random noises. 
We then generate $N=10^5$ testing trajectories for MSD and compute the errors between nominal and true states as $e^i_t, t=0,1, ...,\,i=1, 2,...,N.$
For the SP environment, we only generate $N=100$ testing trajectories due to the complexity of the simulation.
We compare the minimum containment probability $\min_t \mathrm{Pr}\{e_t\in\mathcal{E}_t\}$, which is empirically estimated using $N$ samples.
Moreover, we also compare confidence bound sizes ($\sup_{e\in \mathcal{E}_t}a^\top e$) with baselines and quantiles from samples, where $a$ is the normal of 
the closest constraint boundary of the environment.

\begin{table}[t]
\centering
\caption{Comparison of minimum containment probability over time between different approaches with $95\%$ confidence.
}

\begin{tabular}{cccc}
\hline
  & \textbf{Sub-Gau}     & Robust & DR  \\ \hline
     MSD         & 99.47  &  100.00 & 100.00 \\ 
    SP           & 99  & 100 & 100 \\ 
\hline
\end{tabular}
\label{tab:up}
\end{table}

\revision{The results in \Cref{tab:up} demonstrate that the confidence bounds from sub-Gaussian propagation satisfy the predefined confidence level for heteroscedastic noise in \revision{MSD}.} 
 \Cref{fig:up} illustrates that the bound size from our approach is always greater than the quantile bounds from samples and smaller than the robust and DR bounds, highlighting reliability and reduced conservatism of the uncertain prediction.

\subsection{Stochastic MPC}
\label{sec:exp_mpc}

In this section, we evaluate the effectiveness of our approach for output-feedback stochastic MPC. 
This approach is compared against the distributional robust MPC \cite{hewing2020recursively}.
Robust MPC is not compared with other approaches since it fails to find feasible solutions for all testing environments, which is due to the significantly larger PRS shown in~\Cref{fig:up}. 
In the MSD environment, we utilize the half-space confidence bounds (\Cref{lemma:half_space}) for all stochastic MPC approaches.
Elliptical bounds are used for the SP environment.
The evaluation metrics include the total cost and maximum constraint violation ratio through time, measured over 100 trajectories.

The results in \Cref{tab:mpc} show the capability of our approach to satisfy the chance constraints while being less conservative than the distributional robust approaches.
In \Cref{tab:mpc}, our satisfaction of the chance constraints are all greater than $95\%$, the desired value.
Our average costs are smaller than those of the variance-based distributional robust approach.
Finally, \Cref{fig:mpc_msd,fig:mpc_vl_sp} show the confidence sets from our sub-Gaussian approach, which are reasonably small for finding feasible solutions to the considered problems, including SP with vision-based state estimation.

\begin{table}[!hbtp]
\footnotesize
\centering
\caption{\revision{Performance for MPC approaches in MSD and SP environments.}
\revision{MCP denotes the maximum constraint violation probability, which should be smaller than $\delta=5\%$.}
}
\begin{tabular}{ccccc}
\hline
Envs  & \multicolumn{2}{c}{MSD (Student-t)}  & \multicolumn{2}{c}{SP (Bounded Laplace)} \\ \cmidrule(lr){2-3}\cmidrule(lr){4-5} 
Metrics    &    MCP {[}\%{]} & Cost    & MCP {[}\%{]}     & Cost       \\ \hline
\textbf{Sub-Gau}     & 2            & 59.7                  & 1              & 24.985  \\
DR                & 1            & 64.7       & 0              & 24.986    \\ \hline
\end{tabular}
\label{tab:mpc}
\end{table}

\section{Conclusion}

In this work, we proposed a guaranteed stochastic uncertainty propagation framework based on an extended sub-Gaussian definition.
We derived sub-Gaussian characterization and confidence bounds for the state distribution resulting from sub-Gaussian noise.
We validated our theoretical contributions through sufficient numerical evaluation of our method, demonstrating its capability to guarantee chance constraint satisfaction while being less conservative than robust and distributional robust
approaches.
Interesting future directions include 
extending the stochastic MPC to nonlinear systems and leveraging the sub-Gaussian characterization in machine learning.

\bibliographystyle{plain}        
\bibliography{refs}

\begin{thebibliography}{10}

\bibitem{ao2025saferplan}
Yunke Ao, Hooman Esfandiari, Fabio Carrillo, Christoph~J Laux, Yarden As, Ruixuan Li, Kaat Van~Assche, Ayoob Davoodi, Nicola~A Cavalcanti, Mazda Farshad, et~al.
\newblock Saferplan: Safe deep reinforcement learning for intraoperative planning of pedicle screw placement.
\newblock {\em Medical Image Analysis}, 99:103345, 2025.

\bibitem{aolaritei2023wasserstein}
Liviu Aolaritei, Marta Fochesato, John Lygeros, and Florian D{\"o}rfler.
\newblock Wasserstein tube mpc with exact uncertainty propagation.
\newblock In {\em 2023 62nd IEEE Conference on Decision and Control (CDC)}, pages 2036--2041. IEEE, 2023.

\bibitem{arbel2020strict}
Julyan Arbel, Olivier Marchal, and Hien~D Nguyen.
\newblock On strict sub-gaussianity, optimal proxy variance and symmetry for bounded random variables.
\newblock {\em ESAIM: Probability and Statistics}, 24:39--55, 2020.

\bibitem{bemporad2007robust}
Alberto Bemporad and Manfred Morari.
\newblock Robust model predictive control: A survey.
\newblock In {\em Robustness in identification and control}, pages 207--226. Springer, 2007.

\bibitem{cambanis1981theory}
Stamatis Cambanis, Steel Huang, and Gordon Simons.
\newblock On the theory of elliptically contoured distributions.
\newblock {\em Journal of Multivariate Analysis}, 11(3):368--385, 1981.

\bibitem{chou2022safe}
Glen Chou, Necmiye Ozay, and Dmitry Berenson.
\newblock Safe output feedback motion planning from images via learned perception modules and contraction theory.
\newblock In {\em International Workshop on the Algorithmic Foundations of Robotics}, pages 349--367. Springer, 2022.

\bibitem{chowdhury2017kernelized}
Sayak~Ray Chowdhury and Aditya Gopalan.
\newblock On kernelized multi-armed bandits.
\newblock In {\em International Conference on Machine Learning}, pages 844--853. PMLR, 2017.

\bibitem{drews2019vision}
Paul Drews, Grady Williams, Brian Goldfain, Evangelos~A Theodorou, and James~M Rehg.
\newblock Vision-based high-speed driving with a deep dynamic observer.
\newblock {\em IEEE Robotics and Automation Letters}, 4(2):1564--1571, 2019.

\bibitem{farina2015approach}
Marcello Farina, Luca Giulioni, Lalo Magni, and Riccardo Scattolini.
\newblock An approach to output-feedback mpc of stochastic linear discrete-time systems.
\newblock {\em Automatica}, 55:140--149, 2015.

\bibitem{farina2016stochastic}
Marcello Farina, Luca Giulioni, and Riccardo Scattolini.
\newblock Stochastic linear model predictive control with chance constraints--a review.
\newblock {\em Journal of Process Control}, 44:53--67, 2016.

\bibitem{hewing2020recursively}
Lukas Hewing, Kim~P Wabersich, and Melanie~N Zeilinger.
\newblock Recursively feasible stochastic model predictive control using indirect feedback.
\newblock {\em Automatica}, 119:109095, 2020.

\bibitem{hewing2019scenario}
Lukas Hewing and Melanie~N Zeilinger.
\newblock Scenario-based probabilistic reachable sets for recursively feasible stochastic model predictive control.
\newblock {\em IEEE Control Systems Letters}, 4(2):450--455, 2019.

\bibitem{kohler2024predictive}
Johannes K{\"o}hler and Melanie~N Zeilinger.
\newblock Predictive control for nonlinear stochastic systems: Closed-loop guarantees with unbounded noise.
\newblock {\em IEEE Transactions on Automatic Control}, 2025.

\bibitem{kwakernaak1972linear}
H.~Kwakernaak and R.~Sivan.
\newblock {\em Linear Optimal Control Systems}.
\newblock Wiley-Interscience, 1972.

\bibitem{langson2004robust}
Wilbur Langson, Ioannis Chryssochoos, SV~Rakovi{\'c}, and David~Q Mayne.
\newblock Robust model predictive control using tubes.
\newblock {\em Automatica}, 40(1):125--133, 2004.

\bibitem{li2023distributionally}
Bin Li, Tao Guan, Li~Dai, and Guang-Ren Duan.
\newblock Distributionally robust model predictive control with output feedback.
\newblock {\em IEEE Transactions on Automatic Control}, 2023.

\bibitem{li2024distributionally}
Ruiqi Li, John~W Simpson-Porco, and Stephen~L Smith.
\newblock Distributionally robust stochastic data-driven predictive control with optimized feedback gain.
\newblock {\em arXiv preprint arXiv:2409.05727}, 2024.

\bibitem{li2015vision}
Zhijun Li, Chenguang Yang, Chun-Yi Su, Jun Deng, and Weidong Zhang.
\newblock Vision-based model predictive control for steering of a nonholonomic mobile robot.
\newblock {\em IEEE Transactions on Control Systems Technology}, 24(2):553--564, 2015.

\bibitem{lindemann2024formal}
Lars Lindemann, Yiqi Zhao, Xinyi Yu, George~J Pappas, and Jyotirmoy~V Deshmukh.
\newblock Formal verification and control with conformal prediction.
\newblock {\em arXiv preprint arXiv:2409.00536}, 2024.

\bibitem{mammarella2022chance}
Martina Mammarella, Victor Mirasierra, Matthias Lorenzen, Teodoro Alamo, and Fabrizio Dabbene.
\newblock Chance-constrained sets approximation: A probabilistic scaling approach.
\newblock {\em Automatica}, 137:110108, 2022.

\bibitem{mark2021data}
Christoph Mark and Steven Liu.
\newblock Data-driven distributionally robust mpc: An indirect feedback approach.
\newblock {\em arXiv preprint arXiv:2109.09558}, 2021.

\bibitem{mayne2006robust}
David~Q Mayne, Sa{\v{s}}a~V Rakovi{\'c}, Rolf Findeisen, and Frank Allg{\"o}wer.
\newblock Robust output feedback model predictive control of constrained linear systems.
\newblock {\em Automatica}, 42(7):1217--1222, 2006.

\bibitem{muntwiler2023lqg}
Simon Muntwiler, Kim~P Wabersich, Robert Miklos, and Melanie~N Zeilinger.
\newblock Lqg for constrained linear systems: Indirect feedback stochastic mpc with kalman filtering.
\newblock In {\em 2023 European Control Conference (ECC)}, pages 1--7. IEEE, 2023.

\bibitem{nolan2013multivariate}
John~P Nolan.
\newblock Multivariate elliptically contoured stable distributions: theory and estimation.
\newblock {\em Computational statistics}, 28:2067--2089, 2013.

\bibitem{paulson2020stochastic}
Joel~A Paulson, Edward~A Buehler, Richard~D Braatz, and Ali Mesbah.
\newblock Stochastic model predictive control with joint chance constraints.
\newblock {\em International Journal of Control}, 93(1):126--139, 2020.

\bibitem{Prandini2012scenario}
Maria Prandini, Simone Garatti, and John Lygeros.
\newblock A randomized approach to stochastic model predictive control.
\newblock In {\em 2012 IEEE 51st IEEE Conference on Decision and Control (CDC)}, pages 7315--7320, 2012.

\bibitem{QIN2003733}
S.Joe Qin and Thomas~A. Badgwell.
\newblock A survey of industrial model predictive control technology.
\newblock {\em Control Engineering Practice}, 11(7):733--764, 2003.

\bibitem{rawlings2017model}
J.B. Rawlings, D.Q. Mayne, and M.~Diehl.
\newblock {\em Model Predictive Control: Theory, Computation, and Design}.
\newblock Nob Hill Publishing, 2017.

\bibitem{richards2005robust}
A.~Richards and J.~How.
\newblock Robust model predictive control with imperfect information.
\newblock In {\em Proceedings of the 2005, American Control Conference, 2005.}, pages 268--273 vol. 1, 2005.

\bibitem{sarker2023accurate}
Arnab Sarker, Peter Fisher, Joseph~E Gaudio, and Anuradha~M Annaswamy.
\newblock Accurate parameter estimation for safety-critical systems with unmodeled dynamics.
\newblock {\em Artificial Intelligence}, 316:103857, 2023.

\bibitem{swihart2022multivariate}
Bruce~J Swihart and John~P Nolan.
\newblock Multivariate subgaussian stable distributions in r.
\newblock {\em R Journal}, 14(3), 2022.

\bibitem{vershyninHighdimensionalProbabilityIntroduction2018}
Roman Vershynin.
\newblock {\em High-Dimensional Probability: An Introduction with Applications in Data Science}.
\newblock Number~47 in Cambridge Series in Statistical and Probabilistic Mathematics. {Cambridge University Press}, 2018.

\end{thebibliography}

\appendix

\section{APPENDIX}
\subsection{Proof of \Cref{lemma:equ}}
\label{sec:proof_equ}

\begin{proof}
    \textit{Definition~\ref{def:classic} to Definition~\ref{def:sub_gau_mean}}:
    Let us assume that $\sigma$ is the variance proxy of $X$ with \Cref{def:classic}.
    This means for all $\|b\|=1$, the scalar random variable $b^\top X$ is $\sigma$-sub-Gaussian.
    Therefore, $\forall \lambda\in \mathbb{R}^n$, $\frac{\lambda^\top (X-\mu)}{\|\lambda\|}$ is $\sigma$-sub-Gaussian.
    Then by \Cref{def:classic}, we have $\forall\lambda\in \mathbb{R}^n$:
\begin{align*}
    &\E\left[\exp{\left(\lambda^\top (X - \mu)\right)}\right] \\
    =& \E\left[\exp{\left(\|\lambda\|\cdot\dfrac{\lambda^\top (X-\mu)}{\|\lambda\|}\right)}\right] 
    \leq \exp{\left(\dfrac{\|\lambda\|^2\sigma^2}{2}\right)}.
\end{align*}
Hence, $X$ is sub-Gaussian (\Cref{def:sub_gau_mean}) with variance proxy $\Sigma=\sigma^2 I$.

\textit{Definition~\ref{def:sub_gau_mean} to Definition~\ref{def:classic}}:
    According to Definition~\ref{def:sub_gau_mean}, there is a variance proxy $\Sigma$ such that for $\forall\,\, \lambda\in \mathbb{R}^n$:
    \begin{align*}
        \E\left[\exp{\left(\lambda^T(X - \mu)\right)}\right] &\leq \exp{\left(\frac{\|\lambda\|^2_\Sigma}{2}\right)}
    \end{align*}
    Hence, for any $c\in \mathbb{R}$ and $\lambda\in\mathbb{R}^n$, we have:
     \begin{align*}
        & \E\left[\exp{\left(c(\lambda^TX - \lambda^T\mu)\right)}\right] 
        = \E\left[\exp{\left(c\lambda^T(X - \mu)\right)}\right] \\
        \leq &\exp{\left(\frac{\|c\lambda\|^2_\Sigma}{2}\right)}
         = \exp{\left(\frac{c^2\|\lambda\|^2_\Sigma }{2}\right)}. 
    \end{align*}
    This means for $\forall\,\, \lambda\in \mathbb{R}^n$, $\lambda^TX$ is sub Gaussian with variance proxy $\|\lambda\|_\Sigma$, which means the random vector $X$ is sub-Gaussian with variance proxy $\sigma^2=\revision{\|\Sigma\|_2}$ by Definition~\ref{def:classic}.
\end{proof}

\subsection{Proof of Theorem~\ref{thr:elliptical}}
\label{sec:proof_tail}

\begin{proof}
This proof follows the steps in \cite[Lemma 2]{chowdhury2017kernelized}. 
Without loss of generality, suppose $E[X]=\mu=0$.
According to Definition~\ref{def:sub_gau_mean}, we have for $\forall \lambda \in \mathbb{R}^n$:
\begin{align*}
   \E\left[\exp{\left(\lambda^TX - \dfrac{\|\lambda\|^2_{\Sigma}}{2}\right)}\right] \leq 1.
\end{align*}
Therefore, for $\lambda$ sampled from any Gaussian distribution $\lambda \sim \mathcal{N}(0, S^{-1})$, we also have:
\begin{align*}
    &\int_\lambda \E_X\left[\exp\left(\lambda^TX- \frac{1}{2}\|\lambda\|^2_{\Sigma}\right)\right] p(\lambda) d\lambda \leq 1.
\end{align*}
Now we compute the left-hand side:
\begin{align*}
    &\int_\lambda \E_X\left[\exp\left(\lambda^TX - \frac{1}{2}\|\lambda\|^2_{\Sigma}\right)\right] p(\lambda) d\lambda \\
    & =  \frac{1}{\sqrt{(2\pi)^n\det{(S^{-1})}}}\E_X\left[\int_\lambda\exp\left(\lambda^TX - \frac{1}{2}\|\lambda\|^2_{\Sigma + S} \right) d\lambda  \right]\\
    & = \frac{1}{\sqrt{(2\pi)^n\det{(S^{-1})}}} \E_X\left[\exp\left(\frac{1}{2}\|X\|^2_{(\Sigma + S)^{-1}}\right)\right.  \\
    & \left.\times \int_\lambda \exp\left( - \frac{1}{2}\|\lambda - (\Sigma + S)^{-1}X\|^2_{\Sigma + S} \right) d\lambda \right]\\
    & =  \sqrt{\frac{\det{S}}{\det{(\Sigma + S)}} }\E\left[\exp\left(\frac{\|X\|^2_{(\Sigma + S)^{-1}}}{2}\right)\right].
\end{align*}
Therefore for any $S\succ 0$, we have:
\begin{align*}
    \E\left[\exp\left(\frac{\|X\|^2_{(\Sigma + S)^{-1}}}{2}\right)\right] \leq \sqrt{\frac{\det{(\Sigma + S)}}{\det{(S)}} }.
\end{align*}
Now let us assign $S = m\Sigma, m>0$, then we obtain:
\begin{align*}
    \E\left[\exp \left(\frac{\|X\|^2_{\Sigma^{-1}}}{2 + 2m}\right)\right] 
    \leq \sqrt{\frac{\det{(1+m)\Sigma}}{\det{(m\Sigma)}} } = \left(\frac{1 + m}{m}\right)^\frac{n}{2}. \label{ineq:moment_2}
\end{align*}
Finally, we get for $\forall m>0$ and $t\geq 0$:
\begin{equation}
    \begin{split}
    &\,\,\mathrm{Pr}\{\|X\|_{\Sigma^{-1}} \geq \tau \} \\
    & = \mathrm{Pr}\left\{\exp{\left(\frac{\|X\|^2_{\Sigma^{-1}}}{2 + 2m}\right)} \geq \exp{\left(\frac{\tau^2}{2 + 2m}\right)} \right\} \\
    & \leq \E\left[\exp{\left(\frac{\|X\|^2_{\Sigma^{-1}}}{2 + 2m}\right)} \right] \cdot \exp{\left(-\frac{\tau^2}{2 + 2m}\right)} \\
    & \leq \left(\frac{1 + m}{m}\right)^\frac{n}{2} \exp{\left(-\frac{\tau^2}{2 + 2m}\right)},
\end{split}\label{ineq:tail_m}
\end{equation}
where the second last inequality is the Chernoff inequality.
Now we minimize this tail bound over $m$:
\begin{align*}
    &\frac{d}{dm}\left(\frac{1 + m}{m}\right)^\frac{n}{2} \exp{\left(-\frac{\tau^2}{2(1 + m)}\right)} = 0 \\
    \Rightarrow\,\,& \left(-\frac{n}{2m^2} + \frac{\tau^2}{2(1 + m)m}\right) = 0 
    \,\,\Rightarrow\,\, m^* = \frac{n}{\tau^2 - n}, 
\end{align*}
where $\tau^2 - n > 0$ by assumption.
Plugging $m^*$ to Inequality \labelcref{ineq:tail_m} yields:
\begin{align*}
    \mathrm{Pr}\{\|X\|_{\Sigma^{-1}} \geq \tau \} \leq \left(\frac{\tau^2}{n}\right)^{\frac{n}{2}} \exp{\left(\dfrac{n - \tau^2}{2}\right)},
\end{align*}
which can be rearranged as \Cref{ineq:tail}.
Abbreviating $s:=\frac{\tau^2}{n}-1$ and assigning the tail probability to $\delta$, we have:
\begin{align}
    \left(\dfrac{\exp(s)}{1+s}\right)^{\frac{n}{2}} = \frac{1}{\delta} ~\Rightarrow~  \dfrac{\exp(s)}{1+s}=\delta^{-\frac{2}{n}}. \label{eq:value_s}
\end{align}
Therefore, $s = g^{-1}\left(\delta^{-\frac{2}{n}}\right)$ and the confidence bound is solved as $\tau^2 = n + ng^{-1}\left(\delta^{-\frac{2}{n}}\right)$ as in \Cref{eq:e_set}.
\end{proof}
\subsection{Proof of \Cref{lemma:complexity}}
\label{sec:proof_complexity}
\begin{proof}
    Denote $s=g^{-1}(\delta^{-\frac{n}{2}})$ and $\tau^2=n(s+1)$.  
    Since $1 + s \leq 2\exp(\frac{s}{2})-1$ for $s\geq 0$,
    we have:
    \begin{align*}
        &\frac{\exp(s)}{2\exp(\frac{s}{2})-1} \leq \frac{\exp(s)}{1 + s}=\delta^{-\frac{2}{n}}\\
        \Leftrightarrow~&\exp(s) - 2\delta^{-\frac{2}{n}}\exp\left(\frac{s}{2}\right) + \delta^{-\frac{2}{n}} \leq 0.
    \end{align*}
    Since the left-hand side is a quadratic function of $\exp(\frac{s}{2})$, it holds:
    \begin{align*}
        \exp(\frac{s}{2}) &\leq \delta^{-\frac{2}{n}} + \sqrt{\delta^{-\frac{4}{n}} - \delta^{-\frac{2}{n}}} \\
        \Rightarrow~{\tau}^2 &\leq n + 2n\ln\left(\delta^{-\frac{2}{n}} + \sqrt{\delta^{-\frac{4}{n}} - \delta^{-\frac{2}{n}}}\right) \\
        & \leq n + 2n\ln\left(2\delta^{-\frac{2}{n}}\right)\\
        & = (1 + 2\ln 2)n + 4\ln \delta^{-1}.\qedhere
    \end{align*}
\end{proof}
\ifautomatica
\else
    \subsection{Proof of \Cref{thr:mpc}}
    \label{sec:proof_mpc}
    \begin{proof}
    The proof follows the arguments of~\cite[Thm.~2]{kohler2024predictive} and ~\cite[Thm.~1]{muntwiler2023lqg}.\\
        \textbf{Recursive feasibility:}
    Given the optimal input $v^*_{0:H-1|t}$ at some time $t$, we assign $v^*_{H|t}:=Kz^*_{H|t}$. 
    For time $t+1$, we consider the candidate inputs $v_{0:H-1|t+1}=v^*_{1:H|t}$, which yields the nominal states $z_{0:H|t+1}=\{z^*_{1:H|t},(A + BK)z^*_{H|t}\}$ using \Cref{eq:policy,eq:z_prop}.
    This is a feasible candidate solution to Problem~\labelcref{prob:mpc} using \Cref{ass:terminal}, $(z_{i|t+1},v_{i|t+1})=(z^*_{i+1|t},v^*_{i+1|t})\in (\mathcal{X}\times \mathcal{U})\ominus \mathcal{E}_{i+t+1}, i \in \{0,1,\dots,H-1\}$, and $z_{H|t+1}=(A+BK)z_{H|t}\in \mathcal{Z}_f$.\\
    \textbf{Chance constraints:} 
    Even though the error $\xi_t$ is not necessarily independent of the MPC input $v_t$, Theorem~\ref{thr:propagation} ensures that $\xi_t\sim\mathcal{SG}\left(0,\Sigma_t^\xi\right)$ and the design of $\mathcal{E}_t$ (Thm.~\ref{thr:elliptical}/Lemma~\ref{lemma:half_space}) ensures $\mathrm{Pr}\{\xi_t\in\mathcal{E}_t\}\geq 1-\delta$, $\forall t\in\mathbb{N}$. 
    Thus, closed-loop constraints satisfaction follows with$(x_t,u_t)=(z_t,v_t) + \xi_t \in (z_t, v_t) \oplus \mathcal{E}_t \subseteq \mathcal{X}\times \mathcal{U}$ from the constraint \labelcref{eq:mpc_cons}. \\
    \textbf{Performance guarantees:}
    We denote
    $u_{i|t}^\star=v^*_{i|t}+K(\bar{x}^\star_{i|t}-z^\star_{i|t}), i=0,\dots H$,
    which satisfies $u_{H|t}^\star=K\bar{x}_{H|t}^\star$.
    The optimal certainty equivalent states $\bar{x}^*_{0:H+1|t}$ are determined by \labelcref{eq:x_init} and $\bar{x}^*_{i+1|t}=A\bar{x}^*_{i|t}+Bu_{i|t}^\star, i=0,\dots H$.
    From \Cref{eq:est}, \labelcref{eq:x_init} and \labelcref{eq:sys}, we have:
    \begin{align}
        &\bar{x}_{0|t+1} = \hat{x}_{t+1}= \bar{x}^*_{1|t} + L(y_{t+1} - C\bar{x}^*_{1|t}) \nonumber\\
        =~ &\bar{x}^*_{1|t} + L\left(C(Ax_t + Bu_t+w_t)+\epsilon_t - CA\hat{x}_t - CBu_t \right)\nonumber\\ 
        =~ &\bar{x}^*_{1|t} + LCA \hat{e}_t + LCw_t + L\epsilon_t =:\bar{x}^*_{1|t} +  \bar{e}_t. \label{eq:err_bar}
    \end{align}
    Using $x_t=\hat{x}_t+\hat{e}_t$, the quadratic stage cost satisfies
    \begin{align}
        \frac{1}{2} \ell(x_t, u_t) \leq \ell(\hat{x}_t, u_t) + \|\hat{e}_t\|^2_Q.\label{ineq:loss_est_true}
    \end{align}
    We denote $\mathcal{J}_H(t)$ as the optimal objective function of Problem \labelcref{prob:mpc} at time $t$.
    Following the arguments in~\cite[Thm 2, proof (i)]{kohler2024predictive} and Equation~\labelcref{eq:err_bar}, the quadratic cost and Lipschitz continuous dynamics ensure
    \begin{align}
        \frac{1}{1 + m} \mathcal{J}_H(t+1) \leq\mathcal{J}_H(t) - \ell(\hat{x}_t, u_t) + \frac{c_\mathcal{J}}{m}\|\bar{e}_t\|^2 \label{ineq:J}
    \end{align}
    for all $m>0$ with a uniform constant $c_\mathcal{J}>0$.
    We now consider the upper bound for $\mathrm{tr}(\Sigma_\infty)$. 
    The variance propagation~\labelcref{eq:prop} and $A^e$ Schur stable imply that:
    \begin{align}
        \mathrm{tr}(\Sigma_\infty) \leq c_1(\sigma_{\epsilon}^2 + \sigma_w^2)\label{ineq:bound_trace_1}
    \end{align}
    for some constant $c_1>0$.
    Furthermore, \Cref{thr:propagation} and  \labelcref{eq:err_bar} ensure $\bar{e}_t\sim\mathcal{SG}(0,\bar{\Sigma}_t)$ $\hat{e}_t\sim\mathcal{SG}(0,\hat{\Sigma}_t)$ with
    \begin{align*}
     \hat{\Sigma}_\infty=&[I;0]\Sigma_\infty[I;0]^\top,\\
     \bar{\Sigma}_\infty=&L(C(A\hat{\Sigma}_\infty A^\top + \sigma_w^2 I)C^\top + \sigma_\epsilon^2I)L^\top.
    \end{align*}
    This further implies:
    \begin{equation}
        \begin{split}
            \mathrm{tr}(\hat{\Sigma}_\infty) &\leq \mathrm{tr}(\Sigma_\infty) \\
            \mathrm{tr}(\bar{\Sigma}_\infty) &\leq c_2(\mathrm{tr}(\Sigma_\infty) + \sigma_\epsilon^2 + \sigma_w^2)\end{split}\label{ineq:bound_trace_2}
    \end{equation}
    for some constant $c_2>0$. Applying \Cref{lemma:mom_bound} with $p=2$ in combination with \labelcref{ineq:bound_trace_1} and \labelcref{ineq:bound_trace_2} implies:
    \begin{align}
        &\E[\|\bar{e}_t\|^2] \leq \mathrm{tr}(\bar{\Sigma}_t)\E[\|\bar{e}_t\|^2_{\bar{\Sigma}_t^{-1}}] \leq \mathcal{B}(2, n_x)\mathrm{tr}(\bar{\Sigma}_t) \nonumber\\
        &\E[\|\hat{e}_t\|^2_Q] \leq \mathcal{B}(2, n_x)\lambda_{\mathrm{max}}(Q)\mathrm{tr}(\hat{\Sigma}_t),
         \label{ineq:err_bound}
    \end{align}
    where we use $\lambda_{max}(\Sigma)\leq \mathrm{tr}(\Sigma)$. %
    Combining~\labelcref{ineq:loss_est_true,ineq:J,ineq:err_bound} yields
    \begin{align*}
    &\mathbb{E}_{\epsilon_t,w_t}\left[\dfrac{1}{1+m}J_H(t+1)-J_H(t)+\frac{1}{2}\ell(x_t,u_t)\right] \\
    \leq &\mathcal{B}(2, n_x)\left(\lambda_{\mathrm{max}}(Q)\mathrm{tr}(\hat{\Sigma}_t)+\frac{c_J}{m}\mathrm{tr}(\bar{\Sigma}_t)\right).   
    \end{align*}
    Finally, following~\cite[(iii), proof Thm.~2]{kohler2024predictive}, we choose $m>0$ sufficiently small to arrive at
    \begin{align*}
        &\quad\lim_{T\rightarrow\infty} \E_{\epsilon_{0:T},w_{0:T}}\left[\frac{1}{T}\sum_{t=0}^{T-1}\ell(x_t, u_t)\right]\\ &\leq \lim_{T\rightarrow\infty}\frac{1}{T}\sum_{t=0}^{T-1}\left(\kappa_1\left(\mathrm{tr}(\bar{\Sigma}_t)\right)+\kappa_2\left(\mathrm{tr}(\hat{\Sigma}_t)\right)\right) \\
        &= \kappa_1\left(\mathrm{tr}(\bar{\Sigma}_\infty)\right)+\kappa_2\left(\mathrm{tr}(\hat{\Sigma}_\infty)\right) \\
       &\stackrel{\labelcref{ineq:bound_trace_2},\labelcref{ineq:bound_trace_1}}{\leq}  \kappa_w(\sigma_w) + \kappa_\epsilon(\sigma_\epsilon). \qedhere
    \end{align*}
    
    \end{proof}
\fi

\end{document}